\documentclass[envcountsect, envcountsame]{llncs}

\usepackage{makeidx}
\usepackage{textcomp}
\usepackage{amsmath,amssymb}
\usepackage{graphicx,graphics}
\usepackage{psfrag}
\usepackage{epsfig}

\newcommand{\OR}{\ensuremath{{\rm OR}}}
\newcommand{\val}{\ensuremath{{\rm value}}}

\newcommand{\sat}{\ensuremath{\textnormal{sat}}}

\newcommand{\E}{\ensuremath{\mathbf{E}}}

\newcommand{\luar}{\ensuremath{\leftarrow_{\rm u.a.r.}}}

\newcommand{\ppz}{\texttt{\textup{ppz}}}
\newcommand{\ppsz}{\texttt{\textup{ppsz}}}
\newcommand{\ignore}[1]{}

\usepackage{cancel}
\usepackage{algorithm}
\usepackage{algorithmic}

\begin{document}
\pagestyle{headings}

\title{PPZ For More Than Two Truth Values -- An Algorithm
for Constraint Satisfaction Problems} 

\author{Dominik Scheder}

\institute{Theoretical Computer Science, ETH Z\"urich\\
  CH-8092 Z\"urich, Switzerland\\
  \email{dscheder@inf.ethz.ch}\\ \vspace{5mm}
  \today
}

\maketitle

\begin{abstract}
  We analyze the so-called ppz algorithm for $(d,k)$-CSP problems for
  general values of $d$ (number of values a variable can take) and $k$
  (number of literals per constraint). To analyze its success
  probability, we prove a correlation inequality for submodular
  functions.
\end{abstract}

\section{Introduction}

Consider the following extremely simple randomized algorithm for
$k$-SAT: Pick a variable uniformly at random and call it $x$.  If the
formula $F$ contains the unit clause $(x)$, set $x$ to $1$. If it
contains $(\bar{x})$, set it to $0$.  It if contains neither, set $x$
uniformly at random (and if it contains both unit clauses, give up).
This algorithm has been proposed and analyzed by Paturi, Pudl\'ak,
and Zane~\cite{ppz} and is called ppz.

The idea behind analyzing its success probability can be illustrated
nicely if we assume, for the moment, that $F$ has a unique satisfying
assignment $\alpha$ setting all variables to $1$. Switching a variable
it from $1$ to $0$ makes the formula unsatisfied. Therefore, there is
a clause $C_x = (x \vee \bar{y}_1 \vee \dots \vee\bar{y}_{k-1})$.
With probability $1/k$, the algorithm picks and sets
$y_1,\dots,y_{k-1}$ before picking $x$. Supposed they $y_j$ have been
set correctly (i.e., to $1$), the clause $C_x$ is now reduced to
$(x)$, and therefore $x$ is also set correctly. Intuitively, this
shows that on average, the algorithm has to guess $(1-1/k)n$ variables
correctly and can infer the correct values of the remaining $n/k$
variables. This increases the success probability of the algorithm
from $2^{-n}$ (simple stupid guessing) to $2^{-n(1-1/k)}$.

In this paper we generalize the sketched algorithm to general
constraint satisfaction problems, short CSPs. These are a
generalization of boolean satisfiability to problems involving more
than two truth values.  A set of $n$ variables $x_1,\dots,x_n$ is
given, each of which can take a value from $[d] := \{1,\dots,d\}$.
Each assignment to the $n$ variables can be represented as an element
of $[d]^n$. A {\em literal} is an expression of the form $(x_i \ne c)$
for some $c \in [d]$. A CSP {\em formula} consists of a conjunction
(AND) of {\em constraints}, where a constraint is a disjunction (OR)
of literals.  We speak of $(d,k)$-CSP formula if each constraint
consists of at most $k$ literals. Finally, $(d,k)$-CSP is the problem
of deciding whether a given $(d,k)$-CSP formula has a satisfying
assignment.  Note that $(2,k)$-CSP is the same as $k$-SAT. Also
$(d,k)$-CSP is well-known to be NP-complete, unless $d=1$, $k=1$, or
$d=k=2$.  We can manipulate a CSP formula $F$ by permanently
substituting a value $c$ for a variable $x$. This means we remove all
satisfied constraints, i.e., those containing a literal $(x \ne c')$
for some $c' \ne c$, and from the remaining constraints remove the
literal $(x \ne c)$, if
present. We denote the resulting formula by $F^{[x \mapsto c]}$.

It is obvious how to generalize the algorithm to $(d,k)$-CSP problems.
Again we process the variables in a random order. When picking $x$, we
collect all unit constraints of the form $(x \ne c)$ and call the
value $c$ {\em forbidden}. Values in $[d]$ which are not forbidden are
called {\em allowed}, and we set $x$ to a value that we choose
uniformly at random from all allowed values. How can one analyze the
success probability? Let us demonstrate this for $d=k=3$. Suppose 
$F$ has exactly one satisfying assignment $\alpha=(1,\dots,1)$. Since
changing the value of a variable $x$ from $1$ to $2$ or to $3$ makes
$F$ unsatisfied, we find {\em critical constraints}
\begin{eqnarray*}
  (x \ne 2 \vee y \ne 1 \vee z \ne 1)\\
  (x \ne 3 \vee u \ne 1 \vee v \ne 1)
\end{eqnarray*}
If all variables $y,z,u,v$ are picked before $x$, then there is only
one allowed value for $x$ left, namely $1$, and with probability $1$,
the algorithm picks the correct values. If $y,z$ come before $x$, but
at least one of $u$ or $v$ come after $x$, then it is possible that
the values $1$ and $3$ are allowed, and the algorithm picks the
correct value with probability $1/2$. In theory, we could list all
possible cases and compute their probability. But here comes the
difficulty: The probability of all variables $y,z,u,v$ being picked
before $x$ depends on whether these variables are distinct! Maybe
$y=u$, or $z=v$... For general $d$ and $k$, we get $d-1$ 
critical constraints
\begin{eqnarray}
  C_{2} & := & (x \ne 2 \vee y_{1}^{(2)}\ne 1 \vee \dots
  \vee y_{k-1}^{(2)} \ne 1)\nonumber \\
  C_{3} & := & (x \ne 3 \vee y_{1}^{(3)} \ne 1\vee \dots
  \vee y_{k-1}^{(3)} \ne 1)\nonumber \\
  & \dots & \label{begin-critical-clauses} \\
  C_{d} & := & (x \ne d  \vee  y_{1}^{(d)}\ne 1 \vee  \dots
   \vee  y_{k-1}^{(d)} \ne 1) \ . \nonumber 
\end{eqnarray}
We are interested in the distribution of the number of allowed values
for $x$. However, the above constraints can intersect in complicated
ways, since we have no guarantee that the variables $y^{(c)}_j$ are
distinct. Our main technical contribution is a sort of correlation
lemma showing that in the worst case, the $y^{(c)}_j$ are indeed
distinct, and therefore we can focus on that case, which we are able
to analyze.
\subsection*{Previous Work}
Feder and Motwani~\cite{feder-motwani} were the first to generalize
the ppz-algorithm to CSP problems. In their paper, they consider
$(d,2)$-CSP problem, i.e., each variable can take on $d$ values, and
every constraint has at most two literals.  In this case, the clauses
$C_2,\dots,C_d$ cannot form complex patterns. Feder and Motwani show
that the worst case happens if (i) the variables $y_1^{(2)},\dots,
y_1^{(d)}$ are pairwise distinct and (ii) the CSP formula has a unique
satisfying assignment. However,
their proofs do not directly generalize to higher values of $k$.

Recently, Li, Li, Liu, and Xu~\cite{lililiuxu} analyzed ppz for
general CSP problems (i.e., $d,k \geq 3$). Their analysis is overly
pessimistic, though, since they distinguish only the following two
cases, for each variable $x$: When ppz processes $x$, then either (i)
all $d$ values are allowed, or (ii) at least one value is forbidden.
In case (ii), ppz chooses one value randomly from at most $d-1$
values.  Since case (ii) happens with some reasonable probability,
this gives a better success probability than the trivial $d^{-n}$.
However, the authors ignore the case that two, three, or more values
are forbidden and lump it together with case (ii). Therefore, their
analysis does not capture the full power of ppz.

\ignore{ 
 Note that for
$d=2$ no such problem arises, since for each variable we consider only
one such critical constraint. For the case $d \geq 3$ and $k=2$, Feder
and Motwani~\cite{feder-motwani} proved that in the worst case the variables
$y^{(c)}_j$ are distinct. However, for $k \geq 3$, the constraints
above can intersect in much more complicated ways than for $k=2$, thus
the need for a more general correlation lemma.
}
\subsection*{Our Contribution}
Our contribution is to show that ``everything works as expected'',
i.e., that in the worst case all variables $y^{(c)}_j$ in
(\ref{begin-critical-clauses}) are distinct and the formula has a
unique satisfying assignment. For this case, we can compute (or at
least, bound from below) the success probability of the algorithm.
\begin{theorem}
  For $d,k \geq 1$, define 
  $$  
  G(d,k) := \sum_{j=0}^{d-1} \log_2(1+j){{d-1} \choose j}
  \int_0^1 (1-r^{k-1})^j (r^{k-1})^{d-1-j} dr \ .
  $$
  Then there is a randomized algorithm running in polynomial time
  which, given a $(d,k)$-CSP formula over $n$ variables, returns a
  satisfying assignment with probability at least $2^{-nG(d,k)}$.
\label{main-theorem}
\end{theorem}
The algorithm we analyze in this paper is not novel.  It is a
straightforward generalization of the $\ppz$ algorithm to CSP problems
with more than two truth values. However, its analysis is
significantly more difficult than for $d=2$ (and also more difficult
than for large $d$ and $k=2$, the case Feder and
Motwani~\cite{feder-motwani} investigated).

\subsection*{Comparison}

We compare the success probability of Sch\"oning's random walk
algorithm with that of ppz. For ppz, we state the bound given by Li,
Li, Liu, and Xu~\cite{lililiuxu} and by this paper. All bounds are
approximate and ignore polynomial factors.

\begin{center}
\begin{tabular}{c||c|c|c}
  $(d,k)$ & Sch\"oning~\cite{Schoening99}\ & 
  Li, Li, Liu, and Xu~\cite{lililiuxu}\ & 
  this paper\\ \hline 
  $(2,3)$ & $1.334^{-n}$ & $1.588^{-n}$  & $1.588^{-n}$ \\ \hline
  $(3,3)$ & $2^{-n}$   & $2.62^{-n}$ & $2.077^{-n}$ \\ \hline
  $(5,4)$ & $3.75^{-n}$  & $4.73$  & $3.672^{-n}$ \\ \hline
  $(6,4)$ & $4.5^{-n}$ & $5.73^{-n}$ & $4.33^{-n}$
\end{tabular}
\end{center}

For small values of $d$, in particular for the boolean case $d=2$,
Sch\"oning's random walk algorithm is much faster than ppz, but ppz
overtakes Sch\"oning already for moderately large values of $d$ and
thus is, to our knowledge, the currently fastest algorithm for
$(d,k)$-CSP.

\section{The Algorithm}

The algorithm itself is simple. It processes the variables
$x_1,\dots,x_n$ according to some random permutation $\pi$. When the
algorithm processes the variable $x$, it collects all unit constraints
of the form $(x \ne c)$ and calls $c$ {\em forbidden}. A truth value
$c$ that is not forbidden is called {\em allowed}.  If the formula is
satisfiable when the algorithm processes $x$, there is obviously at
least one allowed value.  The algorithm chooses uniformly at random an
allowed value $c$ and sets $x$ to $c$, reducing the formula. Then it
proceeds to the next variable. For technical reasons, we think of the
permutation $\pi$ as part of the input to the algorithm, and sampling
$\pi$ uniformly at random from all $n!$ permutations before calling
the algorithm. The algorithm is described formally in
Algorithm~\ref{ppz}.
\begin{algorithm}
\caption{\ppz($F$: a $(d,k)$-CSP formula 
  over variables $V:=\{x_1,\dots,x_n\}$, $\pi$: a permutation of $V$)}
\label{ppz}
\begin{algorithmic}[1]
  \STATE $\alpha :=$ the empty assignment
  \FOR{$i=1,\dots,n$}
  \STATE $x := x_{\pi(i)}$
  \STATE $S(x,\pi) := \{c \in [d] \ | \ (x \ne c) \not \in F\}$
  \IF {$S(x,\pi) = \emptyset$} \RETURN \texttt{failure}
  \ENDIF
  \STATE $b \luar S(x,\pi)$ \label{step-b}
  \STATE $\alpha := \alpha \cup [x \mapsto b]$
  \STATE $F:= F^{[x \mapsto b]}$
  \ENDFOR
  \IF{$\alpha$ satisfies $F$} \RETURN $\alpha$
  \ELSE 
  \RETURN \texttt{failure}
  \ENDIF
\end{algorithmic}
\end{algorithm}
To analyze the success probability of the algorithm, we can assume
that $F$ is satisfiable, i.e. the set $\sat(F)$ of satisfying
assignments is nonempty. This is because if $F$ is unsatisfiable, the
algorithm always correctly returns $\texttt{failure}$. For a fixed
satisfying assignment, we will bound the probability
\begin{eqnarray}
\Pr[\ppz(F,\pi) \textnormal{ returns } \alpha] \ ,
\label{prob-alpha-returned}
\end{eqnarray}
where the probability is over the choice of $\pi$ and over the
randomness used by $\ppz$. The overall success probability
is given by
\begin{eqnarray}
\Pr[\ppz(F,\pi) \textnormal{ is successful}]
= \sum_{\alpha \in \sat_V(F)} \Pr[\ppz(F,\pi) \textnormal{ returns }\alpha] \ .
\label{prob-success}
\end{eqnarray}
In the next section, we will bound (\ref{prob-alpha-returned}) from
below. The bound depends on the level of {\em isolatedness} of
$\alpha$: If $\alpha$ has many satisfying neighbors, its probability to
be returned by $\ppz$ decreases. However, the existence of many
satisfying assignments will in turn increase the sum in
(\ref{prob-success}). In the end, it turns out that the worst case
happens if $F$ has a unique satisfying assignment.  Observe that for
the $\ppz$-algorithm in the boolean case~\cite{ppz}, the unique
satisfiable case is also the worst case, whereas for the improved
version $\ppsz$~\cite{ppsz}, it is not, or at least not known to be.

\section{Analyzing the Success Probability}

\subsection{Preliminaries}

In this section, fix a satisfying assignment $\alpha$. For simplicity,
assume that $\alpha = (1,\dots,1)$, i.e. it sets every variable to
$1$. What is the probability that $\ppz$ returns $\alpha$? For a
permutation $\pi$ and a variable $x$, let $\beta$ be the partial truth
assignment obtained by restricting $\alpha$ to the variables that come
{\em before} $x$ in $\pi$, and define
$$
S(x,\pi,\alpha) := \{c \in [d] \ | \ (x \ne c) \not \in F^{[\beta]} \} \ .
$$
In words, we process the variables according to $\pi$ and set them
according to $\alpha$, but stop before processing $x$. We check which
truth values are not  forbidden for $x$ by a unit constraint,
and collect theses truth values in the set $S(x, \pi, \alpha)$. Let 
us give an example:

\paragraph{Example.} Let $d=3, k=2$, and $\alpha =(1,\dots,1)$. We consider
$$
F = (x \ne 2 \vee y \ne 1) \wedge (x \ne 3 \wedge z \ne 1) \ .
$$
For $\pi=(x,y,z)$, no value is forbidden when processing $x$, thus
$S(x,\pi,\alpha)=\{1,2,3\}$. For $\pi'=(y,x,z)$, then we consider the
partial assignment that sets $y$ to $1$, obtaining
$$
F^{[y \mapsto 1]} = (x \ne 2) \wedge (x \ne 3 \vee z \ne 1) \ ,
$$
and  $S(x,\pi',\alpha)=\{1,3\}$.
Last, for $\pi'' = (y,z,x)$, then we set $y$ and $z$ to $1$,
obtaining
$$
F^{[y \mapsto 1, z \mapsto 1]} = (x \ne 2) \wedge (x \ne 3) \ ,
$$
thus $S(x,\pi'',\alpha)= \{1\}$.$\hfill\Box$\\

Observe that $S(x,\pi,\alpha)$ is non-empty, since $\alpha(x) \in
S(x,\pi,\alpha)$, i.e. the value $\alpha$ assigns to $x$ is always
allowed. What has to happen in order for the algorithm to return
$\alpha$? In every step of $\ppz$, the value $b$ selected in
Line~\ref{step-b} for variable $x$ must be $\alpha(x)$.  Assume now
that this was the case in each of the first $i$ steps of the
algorithm, i.e., the variables $x_{\pi(1)},\dots,x_{\pi(i)}$ have been
set to their respective values under $\alpha$.  Let $x = x_{\pi(i+1)}$
be the variable processed in step $i+1$.  The set $S(x,\pi,\alpha)$
coincides with the set $S(x,\pi)$ of the algorithm, and therefore $x$
is set to $\alpha(x)$ with probability $1 / |S(x,\pi,\alpha)|$. Since
this holds in every step of the algorithm, we conclude that for a
fixed permutation $\pi$,
$$
\Pr[\ppz(F,\pi) \textnormal{ returns } \alpha] 
= \prod_{x \in V} \frac{1}{|S(x,\pi,\alpha)|} \ .
$$
For $\pi$ being chosen uniformly at random, we obtain 
$$
\Pr[\ppz(F,\pi) \textnormal{ returns } \alpha] 
= \E_{\pi}\left[\prod_{x \in V}^n
  \frac{1}{|S(x,\pi,\alpha)|}\right] \ .
$$
The expectation of a product is an uncomfortable term if the factors
are not independent.  The usual trick in this context is to apply
Jensen's inequality, hoping that we do not lose too much.
\begin{lemma}[Jensen's Inequality]
  Let $X$ be a random variable and $f: \mathbb{R} \rightarrow
  \mathbb{R}$ a convex function. Then $\E[f(X)] \geq f(\E[X])$,
  provided both expectations exist.
\end{lemma}
We apply Jensen's inequality with the convex function being $f: x
\mapsto 2^{-x}$ and the random variable being $X =
\sum_{x \in V} \log_2 |S(x,\pi,\alpha)|$. With this notation, $f(X) =
\prod_{x \in V}^n \frac{1}{|S(x,\pi,\alpha)|}$, the expectation of
which we want to bound from below.
\begin{eqnarray}
\E\left[\prod_{x \in V} \frac{1}{|S(x,\pi,\alpha)|}\right] 
& = & \E\left[2^{-\sum_{x \in V} \log_2 |S(x,\pi,\alpha)|}\right] \nonumber \\
& \geq & 
2^{E[-\sum_{x \in V} \log_2 |S(x,\pi,\alpha)|]} \label{good-jensen}\\
& = & 2^{-\sum_{x \in V} E[\log_2 |S(x,\pi,\alpha)|]} \nonumber \ .
\end{eqnarray}
\begin{proposition}
  $\Pr[\ppz(F,\pi) \textnormal{ returns } \alpha] 
  \geq 2^{-\sum_{x \in V} E[\log_2 |S(x,\pi,\alpha)|]}$.
\label{prop-success}
\end{proposition}

\paragraph{Example: The boolean case.} 
In the boolean case, the set $S(x,\pi,\alpha)$ is either $\{1\}$ or
$\{0,1\}$, and thus the logarithm is either $0$ or $1$. Therefore, the
term $E[\log_2 |S(x,\pi,\alpha)|]$ is the probability that the value
of $x$ is not determined by a unit clause, and thus has to be guessed.

So far the calculations are exactly as in the boolean $\ppz$.  This
will not stay that way for long. In the boolean case, there are only
two cases: Either the value of $x$ is determined by a unit clause
(in which we call $x$ {\em forced}), or it is not.  For $d \geq 3$,
there are more cases: The set of potential values for $x$ can be the
full range $[d]$, it can be just the singleton $\{1\}$, but it can
also be anything in between, and even if the algorithm cannot
determine the value of $x$ by looking at unit clauses, it will still
be happy if at least, say, $d/2$ values are forbidden by unit
clauses.

\subsection{Analyzing $E[\log_2 |S(x,\pi,\alpha)|]$}

In this section we prove an upper bound on $E[\log_2
|S(x,\pi,\alpha)|]$. We assume without loss of generality that $\alpha
= (1,\dots,1)$. There are $d$ truth assignments
$\alpha_1,\dots,\alpha_d$ agreeing with $\alpha$ on the variables $V
\setminus \{x\}$: For a value $c \in [d]$ we define $\alpha_c :=
\alpha[x \mapsto c]$, i.e., we change the value it assignment to $x$
to $c$, but keep all other variables fixed.  Clearly, $\alpha_1 =
\alpha$. The number of assignments among $\alpha_1,\dots,\alpha_d$
that satisfy $F$ is called the {\em looseness} of $\alpha$ at $x$,
denoted by
$$
\ell(\alpha,x) \ .
$$
Since $\alpha_1 = \alpha$ satisfies $F$, the looseness of $\alpha$ at
$x$ is at least $1$, and since there are $d$ possible values for $x$,
the looseness is at most $d$. Thus $1 \leq \ell(\alpha,x) \leq d$. If
$\alpha$ is the unique satisfying assignment, then $\ell(\alpha,x) =
1$ for every $x$.  Note that $\alpha$ being unique is sufficient, but
not necessary: Suppose $\alpha= (1,\dots,1)$ and
$\alpha'=(2,2,1,1,\dots,1)$ are the only two satisfying assignments.
Then $\ell(\alpha,x)=\ell(\alpha',x)=1$
for every variable $x$.

Why are we considering the looseness $\ell$ of $\alpha$ at $x$?
Suppose without loss of generality  that the assignments
$\alpha_1,\dots,\alpha_{\ell}$ satisfy $F$, whereas
$\alpha_{\ell+1},\dots,\alpha_d$ do not. The set $S(x,\pi,\alpha)$
is a random object depending on $\pi$, but one thing is sure:
$$
\textnormal{for all } c=1,\dots,\ell(\alpha,x): \ 
c \in S(x,\pi,\alpha) \ .
$$
For $\ell(\alpha,x) < c \leq d$, what is the probability that $c \in
S(x,\pi,\alpha)$? Since $\alpha_c$ does not satisfy $F$, there must be
a constraint in $F$ that is satisfied by $\alpha$ but not by
$\alpha_c$. Since $\alpha$ and $\alpha_c$ disagree on $x$ only, that
constraint must be of the following form:
\begin{eqnarray}
(x \ne c \vee y_2 \ne 1 \vee y_3 \ne 1 \vee \dots \vee y_k \ne 1) \ .
\label{critical-constraint}
\end{eqnarray}
For some $k-1$ variables $y_2,\dots,y_k$. We do not rule out
constraints with fewer than $k-1$ literals, but we capture this
by not insisting on the $y_j$ in (\ref{critical-constraint}) being
distinct. In any case, if the variables $y_2,\dots,y_k$ come before
$x$ in the permutation $\pi$, then $c \not \in S(x,\pi,\alpha)$: This
is because after setting to $1$ the variables that come before $x$,
the constraint in (\ref{critical-constraint}) has been reduced to $(x
\ne c)$. Note that $y_2,\dots,y_k$ coming before $x$ is sufficient for
$c \not \in S(x,\pi,\alpha)$, but not necessary, since there could be
multiple constraints of the form (\ref{critical-constraint}).
With probability at least $1/k$, all variables $y_2,\dots,y_k$
come before $x$, and we conclude:
\begin{proposition}
  If $\alpha_c$ does not satisfy $F$, then 
  $\Pr[ c \in S(x,c,\alpha) ] \leq 1 - 1/k$.
\end{proposition}
This proposition is nice, but not yet useful on its own. We can use it
to finish the analysis of the running time, however we will
end up with a suboptimal estimate.
\subsection{A suboptimal analysis of $\ppz$}
The function $t \mapsto \log_2(t)$ is concave. We apply Jensen's
inequality to conclude that
\begin{eqnarray}
\E[\log_2 |S(x,\pi,\alpha)|] & \leq & 
\log_2\left( \E[|S(x,\pi,\alpha)|]\right)
 = 
\log_2\left(\sum_{c=1}^n \Pr[c \in S(x,\pi,\alpha)]\right)
\label{bad-jensen}
\end{eqnarray}
We apply what we have learned above: For $c = 1,\dots,\ell(\alpha,x)$,
it always holds that $c \in S(x,\pi,\alpha)$, and for $c =
\ell(\alpha,x)+1,\dots,d$, we have computed that $\Pr[c \in
S(x,\pi,\alpha)] \leq 1 - 1/k$. Therefore
\begin{eqnarray*}
  \E[\log_2 |S(x,\pi,\alpha)|] \leq \log_2\left(\ell(\alpha,x) + 
    (d-\ell(\alpha,x))\left(1 - \frac{1}{k}\right)\right) \ .
\end{eqnarray*}
\textbf{The unique case.} If $\alpha$ is the unique satisfying
assignment, then $\ell(\alpha,x)=1$ for every variable $x$ in our CSP
formula $F$, and the above term becomes
$$
\log_2 \left(1 + \frac{(d-1)(k-1)}{k}\right)
= \log_2 \left(\frac{d(k-1)+1}{k}\right) \ . 
$$
We plug this into the bound of 
Proposition~\ref{prop-success}:
\begin{eqnarray*}
\Pr[\ppz \textnormal{ returns } \alpha] 
& \geq & 2^{-\sum_{i=1}^n E[\log_2 |S(x_i,\pi,\alpha)|]}\\
& \geq & 2^{-n \log_2 \left(\frac{d(k-1)+1}{k}\right)}\\
& = & \left(\frac{d(k-1)+1}{k}\right)^{-n} \ .
\end{eqnarray*}
The success probability of Sch\"oning's algorithm for $(d,k)$-CSP
problems is $\left(\frac{d(k-1)}{k}\right)^n$, and we see that even
for the unique case, our analysis of $\ppz$ does not yield anything
better than Sch\"oning. Discouraged by this failure, we do not
continue this suboptimal analysis for the non-unique case.

\subsection{Detour: Jensen's Inequality Here, There, and Everywhere}

The main culprit behind the poor performance of our analysis is
Jensen's inequality in (\ref{bad-jensen}).  To improve our analysis,
we refrain from applying Jensen's inequality there and instead try to
analyze the term $\E[\log_2 |S(x,\pi,\alpha)|]$ directly.  However,
recall that we have used Jensen's inequality before, in
(\ref{good-jensen}). Is it safe to apply it there? How can we tell
when applying it makes sense and when it definitely does not?  To
discuss this issue, we restate the two applications of Jensen's
inequality:
\begin{eqnarray}
  \E\left[2^{-\sum_{x \in V} \log_2 |S(x,\pi,\alpha)|}\right] & \geq & 
  2^{E[-\sum_{x \in V} \log_2 |S(x,\pi,\alpha)|]} \label{2-good-jensen}\\
E[\log_2 |S(x,\pi,\alpha)|] & \leq & 
\log_2\left( \E[|S(x,\pi,\alpha)|]\right)
\label{2-bad-jensen}
\end{eqnarray}
Formally, Jensen's inequality states that for a random variable
$X$ and a convex function $f$, it holds that
\begin{eqnarray}
  \E[f(X)] & \geq & f(\E[X])\label{jensen} \ ,
\end{eqnarray}
and by multiplying (\ref{jensen}) by $-1$ one obtains a similar
inequality for concave functions. As a rule of thumb, Jensen's
inequality is pretty tight if $X$ is very concentrated around its
expectation: In the most extreme case, $X$ is a constant, and
(\ref{jensen}) holds with equality. On the other extreme,
suppose $X$ is a random variable taking on values $-m$ and $m$, each
with probability $1/2$, and let $f: t\mapsto t^2$, which is a convex
function. The left-hand side of (\ref{jensen}) evaluates to $\E[f(X)]
= \E[X^2] = m^2$, whereas the right-hand side evaluates to $f(\E[X]) =
f(0)=0$, and Jensen's inequality is very loose indeed. What random
variables are we dealing with in (\ref{2-good-jensen}) and
(\ref{2-bad-jensen})? These are
\begin{eqnarray*}
  X & := & \sum_{x \in V} \log_2 |S(x,\pi,\alpha)| \qquad \textnormal{and}\\
  Y & := & |S(x,\pi,\alpha)| \ ,
\end{eqnarray*}
and the corresponding functions are $f: t \mapsto 2^{-t}$, which is
convex, and $g: t \mapsto \log_2 t$, which is concave. In both cases,
the underlying probability space is the set of all permutations of
$V$, endowed with the uniform distribution. We see that $Y$ is not
concentrated at all: Suppose $x$ comes first in $\pi$: If our CSP
formula $F$ contains no unit constraints, then $|S(x,\pi,\alpha)| =
d$, i.e., no truth value is forbidden by a unit constraints. On the
other hand, if $x$ comes last in $\pi$, then $|S(x,\pi,\alpha)| =
\ell(\alpha,x)$. Either case happens with probability $1/n$, which is
not very small. Thus, the random variable $|S(x,\pi,\alpha)|$
does not seem to be very concentrated. 

Contrary to $Y$, the random variable $X$ can be very
concentrated, in fact for certain CSP formulas it can be a constant:
Suppose $d=2$, i.e., the boolean case. Here $X$ simply counts the
number of non-forced variables. Consider the $2$-CNF formula
\begin{eqnarray}
\wedge_{i=1}^{n/2} (x_i \vee y_i) \wedge (x_i \vee \bar{y}_i) \wedge 
(\bar{x}_i \vee y_i) \ .
\label{X-is-constant}
\end{eqnarray}
This formula has $n$ variables, and $\alpha = (1,\dots,1)$ is the
unique satisfying assignment. Observe that if $x_i$ comes before $y_i$
in $\pi$, then $S(x_i,\pi,\alpha)=\{0,1\}$ and
$S(y_i,\pi,\alpha)=\{1\}$.  If $y_i$ comes before $x_i$, then
$S(x_i,\pi,\alpha)=\{1\}$ and $S(y_i,\pi,\alpha)=\{0,1\}$.  Hence $X
\equiv n/2$ is a constant. Readers who balk at the idea of supplying a
$2$-CNF formula as an example for an exponential-time algorithm may
try to generalize (\ref{X-is-constant}) for values
of $k \geq 3$.

\subsection{A Better Analysis}\label{subsection-a-better-analysis}

After this interlude on Jensen's inequality, let us try to bound
$\E[\log_2 |S(x,\pi,\alpha)|]$ directly. In this context, $x$ is some
variable, $\alpha$ is a satisfying assignment, for simplicity $\alpha
= (1,\dots,1)$, and $\pi$ is a permutation of the variables sampled
uniformly at random.  Again think of the $d$ truth assignments
$\alpha_1,\dots,\alpha_d$ obtained by setting $\alpha_c := \alpha[x
\mapsto c]$ for $c=1,\dots,d$. Among them, $\ell:=\ell(\alpha,x)$
satisfy the formula $F$. We assume without loss of generality that
those are $\alpha_1, \dots, \alpha_\ell$. Thus, for each $\ell < c
\leq d$, there is a constraint $C_{c}$ satisfied by $\alpha$ but not
by $\alpha_c$.  Let us write down these constraints:
\begin{eqnarray}
  C_{\ell+1} & := & (x \ne \ell+1 \vee y_{1}^{(\ell+1)}\ne 1 \vee \dots
  \vee y_{k-1}^{(\ell+1)} \ne 1)\nonumber \\
  C_{\ell+2} & := & (x \ne \ell+2 \vee y_{1}^{(\ell+2)} \ne 1\vee \dots
  \vee y_{k-1}^{(\ell+2)} \ne 1)\nonumber \\
  & \dots & \label{critical-clauses} \\
  C_{d} & := & (x \ne d \ \ \ \ \ \vee  \ y_{1}^{(d)}\ne 1 \  \ \vee  \dots
   \vee   y_{k-1}^{(d)} \ne 1) \nonumber
\end{eqnarray}
We define binary random variables $Y_{j}^{(c)}$ for $1 \leq j \leq k-1$
and $\ell+1 \leq c \leq d$ as follows:
\begin{eqnarray*}
  Y_j^{(c)} := \left\{ \begin{array}{ll}
      1 &  \textnormal{ if } y_j^{(c)} \textnormal{ comes after } x 
      \textnormal{ in the permutation } \pi \ ,\\ 
      0 &  \textnormal{otherwise} \ .
      \end{array}
\right.
\end{eqnarray*}
We define $Y^{(c)} := Y_1^{(c)} \vee \dots \vee Y_{k-1}^{(c)}$.  For
convenience we also introduce random variables $Y^{(1)},\dots,
Y^{(\ell)}$ that are constant $1$. Finally, we define $Y :=
\sum_{c=1}^d Y^{(c)}$.  Observe that $Y^{(c)} = 0$ if and only if all
variables $y_{1}^{c},\dots,y_{k-1}^{c}$ come before $x$ in the
permutation, in which case $c \not \in S(x,\pi,\alpha)$. Therefore,
\begin{eqnarray}
  |S(x,\pi,\alpha)| \leq Y
\label{S-leq-Y}
\end{eqnarray}
The variables $Y^{(1)},\dots,Y^{(\ell)}$ are constant $1$, whereas
each of the $Y^{(c+1)},\dots,Y^{(d)}$ is $0$ with probability at least
$1/k$. Since $1 \leq \ell \leq d$, the random variable $Y$ can take
values from $1$ to $d$. We want to bound
\begin{eqnarray}
  \E[\log_2|S(x,\alpha,\pi)|] \leq \E[\log_2(Y)]
  = \E\left[\log_2\left( \ell + \sum_{c=\ell+1}^d Y^{(c)}\right)\right]\  .
\label{ineq-S-Y}
\end{eqnarray}
For this, we must bound the probability $\Pr[Y = j]$ for
$j=1,\dots,d$. This is difficult, since the $Y^{(c)}$ are not
independent: For example, conditioning on $x$ coming very early in
$\pi$ increases the expectation of each $Y^{(c)}$, and conditioning on
$x$ coming late decreases it. We use a standard trick, also used by
Paturi, Pud\'ak, Saks and Zane~\cite{ppsz} to overcome these
dependencies: Instead of viewing $\pi$ as a permutation of $V$, we
think of it as a function $V \rightarrow [0,1]$ where for each $x \in
V$, its value $\pi(x)$ is chosen uniformly at random from $[0,1]$.
With probability $1$, all values $\pi(x)$ are distinct and therefore
give rise to a permutation. The trick is that for $x$, $y$, and $z$
being three distinct variables, the events ``y comes before x'' and
``z comes before x'' are independent when conditioning on $\pi(x) =
r$:
\begin{eqnarray*}
  \Pr[\pi(y) < \pi(x) \ | \ \pi(x)=r] & = & r\\
  \Pr[\pi(z) < \pi(x) \ | \ \pi(x)=r] & = & r\\
  \Pr[\pi(x) < \pi(x) \textnormal{ and }
  \pi(z) < \pi(x) \ | \ \pi(x)=r] & = & r^2\\
\end{eqnarray*}
Compare this to the unconditional probabilities:
\begin{eqnarray*}
  \Pr[\pi(y) < \pi(x)] & = & \frac{1}{2}\\
  \Pr[\pi(z) < \pi(x) \ | \ \pi(x)=r] & = & \frac{1}{2}\\
  \Pr[\pi(x) < \pi(x) \textnormal{ and }
  \pi(z) < \pi(x) \ | \ \pi(x)=r] & = & \frac{1}{3}\\
\end{eqnarray*}
We want to compute $\E[Y^{(c)} \ | \ \pi(x)=r]$.  We know that
$\E[Y^{(c)}_j \ | \ \pi(x)=r] = 1-r$, since $Y^{(c)}_j$ is $1$ if and
only if the boolean variable $y^{(c)}_j$ comes {\em after} $x$.  Since
we are dealing with constraints of size at most $k$, there are, for
each $\ell+1 \leq c \leq d$, at most $k-1$ distinct variables
$y^{(c)}_1,\dots,y^{(c)}_{k-1}$, and the probability that all come
before $x$, conditioned on $\pi(x)=r$, is at least $r^{k-1}$.
Therefore
\begin{eqnarray*}
  \E[Y^{(c)}] \leq 1 - r^{k-1} \ .
\end{eqnarray*}
Still, a variable $y^{(c)}_j$ might occur in several constraints among
$C_{\ell+1},\dots,C_d$, and therefore the $Y^{c}$ are not independent.
The main technical tool of our analysis is a lemma stating that the worst
case is achieved exactly if they in fact are independent, i.e., if all
variables $y^{(c)}_j$ for $c = \ell+1,\dots,d$ and $k=1,\dots,k-1$ are
distinct.
\begin{lemma}[Independence is Worst Case]
  Let $r$, $k$, $\ell$ and $Y^{(c)}$ be defined as above. Let
  $Z^{(\ell+1)},\dots,Z^{(d)}$ be independent binary random variables
  with $\E[Z_i] = 1-r^{k-1}$. Then
  $$
 \E\left[\log_2\left(\ell + \sum_{c=\ell+1}^d
    Y^{(c)}\right) \ | \ \pi(x)=r\right] \leq \E\left[\log_2\left(\ell +
  \sum_{c=\ell+1}^{d} Z^{(c)}\right)\right] \ .
  $$
\label{worst-case-independent}
\end{lemma}
Before we prove the lemma in the next section, we first finish the
analysis of the algorithm.  \ignore{To simplify notation, we also
  introduce $d$ independent binary random variables
  $Z^{(1)},\dots,Z^{(d)}$ of expectation $1-r^{k-1}$ instead of the
  $d-\ell$ random variables in Lemma~\ref{worst-case-independent}.
  Since the logarithm is a monotone function, it holds that
\begin{eqnarray}
  \E\left[\log_2\left(\ell +
  \sum_{c=\ell+1}^{d} Z^{(c)}\right)\right] \leq 
  \E\left[\log_2\left(\ell +
  \sum_{c=1}^{d} Z^{(c)}\right)\right]\ .
\end{eqnarray}}
We apply a somewhat peculiar estimate: Let $a \geq 1$ and $b\geq 0$ 
be integers. Then $\log_2(a + b) 
\leq \log_2 (a \cdot (b+1)) = \log_2(a) + \log_2(b+1)$. Applying this
with 
$a := \ell$ and $b := \sum_{c=\ell+1}^d Z^{(c)}$ and combining it 
with the lemma and with (\ref{ineq-S-Y}), we obtain
\begin{eqnarray}
  \E[\log_2|S(x,\alpha,\pi)| \ | \ \pi(x)=r] \leq \log_2(\ell) + 
  \E\left[\log_2 \left(1+\sum_{c=\ell+1}^d Z^{(c)}\right)\right] \ .
\label{ineq-S-Z}
\end{eqnarray}
This estimate looks wasteful, but consider the case where $F$ has a
unique satisfying assignment $\alpha$: There, $\ell(\alpha,x)=1$ for
every variable $x$, and (\ref{ineq-S-Z}) holds with equality.  In
addition to $Z^{(\ell+1)},\dots,Z^{(d)}$, we introduce $\ell-1$ new
independent binary random variables $Z^{(2)},\dots,Z^{(\ell)}$, each
with expectation $1-r^{k-1}$, and define
\begin{eqnarray*}
  g(d,k,r) & := & \E\left[\log_2\left(1+\sum_{c=2}^d Z^{(c)}\right)\right] \ .
\end{eqnarray*}
The only difference between the expectation in (\ref{ineq-S-Z}) and
here is that here, we sum over $c=2,\dots,d$, whereas in
(\ref{ineq-S-Z}) we sum only over$c=\ell+1,\dots,d$. We get the
following version of (\ref{ineq-S-Z}):
\begin{eqnarray}
  \E[\log_2|S(x,\alpha,\pi)| \ \big| \ \pi(x)=r] \leq \log_2(\ell) + 
  g(d,k,r) \ .
\label{ineq-S-f}
\end{eqnarray}
We want to get rid of the condition $\pi(x)=r$. This is done
by integrating (\ref{ineq-S-f}) for $r$ from $0$ to $1$.
\begin{eqnarray}
  \E[\log_2|S(x,\alpha,\pi)|] \leq \log_2(\ell) + 
  \int_0^1 g(d,k,r) dr =: \log_2(\ell) + 
  G(d,k) \ .
\label{ineq-S-F}
\end{eqnarray}
This $G(d,k)$ is indeed the same $G(d,k)$ as in
Theorem~\ref{main-theorem}, and below we will do a detailed
calculation showing this.
\begin{lemma}[Lemma 1 in Feder, Motwani~\cite{feder-motwani}]
  Let $F$ be a satisfiable CSP formula over variable set $V$.  \ignore{For a
  truth assignment $\alpha \in \sat_V(F)$, we define $\val(\alpha) :=
  \prod_{x \in V} 1 / \ell(\alpha,x)$.} Then
  \begin{eqnarray}
    \sum_{\alpha \in \sat_V(F)} \prod_{x \in V} \frac{1}{\ell(\alpha,x)} \geq 1 \ .
  \end{eqnarray}
\label{lemma-kraft}
\end{lemma}

This lemma is a quantitative version of the intuitive statement that
if a set $S \subseteq [d]^n$ is small, then there must be rather
isolated points in $S$. We now put everything together:
\begin{eqnarray*}
  \Pr[\ppsz(F,\pi) \textnormal{ is successful}]
  & = & 
  \sum_{\alpha \in \sat_V(F)} \Pr[\ppsz(F,\pi)\textnormal{ returns } \alpha]\\
  & \geq &
  \sum_{\alpha \in \sat_V(F)} 2^{-\sum_{x \in V}\E[\log_2|S(x,\alpha,\pi)|]} \ ,
\end{eqnarray*}
where the inequality follows from (\ref{good-jensen}). Together
with (\ref{ineq-S-F}), we see that
\begin{eqnarray*}
  \sum_{\alpha \in \sat_V(F)} 2^{-\sum_{x \in V}\E[\log_2|S(x,\alpha,\pi)|]}
  & \geq &
  \sum_{\alpha \in \sat_V(F)} 2^{-\sum_{x \in V}(\log_2(\ell(\alpha,x))+G(d,k))} \\
  & = & 2^{-nG(d,k)} \sum_{\alpha \in \sat_V(F)}
  2^{-\sum_{x \in V} \log_2(\ell(\alpha,x))} \\
  & = & 2^{-nG(d,k)} \sum_{\alpha \in \sat_V(F)} 
  \prod_{x \in V} \frac{1}{\ell(\alpha,x)} \\
  & \geq & 2^{-nG(d,k)} \ ,
\end{eqnarray*}
where the last inequality follows from Lemma~\ref{lemma-kraft}.  To
prove Theorem~\ref{main-theorem}, we evaluate the term $G(d,k)$.
Recall that $G(d,k) = \int_0^1 g(d,k,r) dr$, where $g(d,k,r) =
\E\left[\log_2\left(1+\sum_{c=2}^d Z^{(c)}\right)\right]$, and
$Z^{(2)},\dots,Z^{(d)}$ are independent binary variables with
expectation $1-r^{k-1}$ each. For $0 \leq j \leq d-1$, it holds that
\begin{eqnarray}
  \Pr\left[\sum_{c=2}^d Z^{(c)} = j\right] 
  = {{d-1} \choose j}(1-r^{k-1})^j (r^{k-1})^(d-1-j)
  \  . 
\label{eq-Z}
\end{eqnarray}
By the definition of expectation, it holds that
$$
g(d,k,r) = \sum_{j=0}^{d-1} \log_2(1+j) \Pr\left[\sum_{c=2}^d Z^{(c)} 
= j\right]\ .
$$
Combining this with (\ref{eq-Z}) and integrating over $r$ from $0$ to
$1$ yields the expressions Theorem~\ref{main-theorem}. This finishes
the proof.

\ignore{

\begin{theorem}
  Let $Z^{(2)},\dots,Z^{(d)}$ be independent binary variables with
  expectation $1-r^{k-1}$ each.  Define $G(d,k)  := 
  \int_0^1 \E\left[\log_2\left(1+\sum_{c=2}^d Z^{(c)}\right)\right]dr$. For
  a satisfiable $(d,k)$-CSP formula over $n$ variables, the 
  algorithm $\ppsz$ returns a satisfying truth assignment with probability
  at least $2^{-nG(d,k)}$.
\end{theorem}
}
\section{A Correlation Inequality}

The goal of this section to prove Lemma~\ref{worst-case-independent}.
We will prove a more general statement.
\begin{definition}
  A function $f: \{0,1\}^n \rightarrow \mathbb{R}$ is called
  {\em monotonically increasing}, or simply {\em monotone},
  if for all $\vec{x},\vec{y} \in \{0,1\}^n$ it holds that
  \begin{eqnarray}
    \vec{x} \leq \vec{y} \ \Rightarrow \ f(\vec{x}) \leq f(\vec{y}) \ ,
  \end{eqnarray}
  where $\vec{x} \leq \vec{y}$ is understood pointwise, i.e., $x_i
  \leq y_i$ for all $1 \leq i \leq n$.
\end{definition}
For example, the functions $\wedge$ and $\vee$, seen as functions from
$\{0,1\}^n$ to $\mathbb{R}$, are monotone, whereas the parity function
$\oplus$ is not.
\begin{definition}
  A function $f: \{0,1\}^n \rightarrow \mathbb{R}$ is called
  {\em submodular} if for all $\vec{x},\vec{y} \in \{0,1\}$, it holds
  that
  \begin{eqnarray}
    f(\vec{x}) + f(\vec{y}) \geq f(\vec{x} \wedge \vec{y})  + 
    f(\vec{x} \vee \vec{z}) \ ,
    \label{ineq-submodular}
  \end{eqnarray}
  where $\vee$ and $\wedge$ are understood pointwise, i.e.
  $(x_1,\dots,x_n) \vee (y_1,\dots,y_n) = 
  (x_1 \vee y_1,\dots, x_n \vee y_n)$.\\
\end{definition}
\paragraph{Example.} The OR-function $f: (x_1,\dots,x_n) \mapsto x_1 \vee
\dots \vee x_n$ is monotone and submodular: It is pretty clear that it
is monotone, so let us try to show submodularity.  There are two
cases: First, suppose at least one of $\vec{x}$ and $\vec{y}$ is
$\vec{0}$, say $\vec{y}=\vec{0}$. Then the left-hand side of
(\ref{ineq-submodular}) evaluates to $f(x)$, and the right-hand side
to $f(0) + f(x) = f(x)$. If neither $\vec{x}=\vec{0}$ nor
$\vec{y}=\vec{0}$, then the left-hand side is $2$, and the right-hand
side is obviously at most
$2$.
\paragraph{Example.} The AND-function $g: (x_1,\dots,x_n) \mapsto x_1 \wedge
\dots \wedge x_n$ is monotone, but not submodular. It is clearly
monotone, so let us show that it is not submodular. Consider $n=2$.
Set $\vec{x} = (0,1)$ and $\vec{y} = (1,0)$. Then
$f(\vec{x}) + f(\vec{y}) = 0$, but $f(\vec{x} \wedge \vec{y})
+ f(\vec{x} \vee \vec{y}) = f(0,0) + f(1,1) = 1.$\\

We define the notion of {\em glued restrictions} of functions.  Let
$A, B$ be two arbitrary sets, and let $f: A^n \rightarrow B$ be a
function. We define a new function $f'$ by ``gluing together'' two
input coordinates of $f$. Formally, for $1 \leq i \leq j \leq n$, we
define the function
$$
f': (a_1,\dots,a_n) \mapsto f(a_1,\dots,a_{j-1},
a_i, a_{j+1},\dots,a_n) \ .
$$
\begin{figure}
  \begin{center}
  \epsfig{file=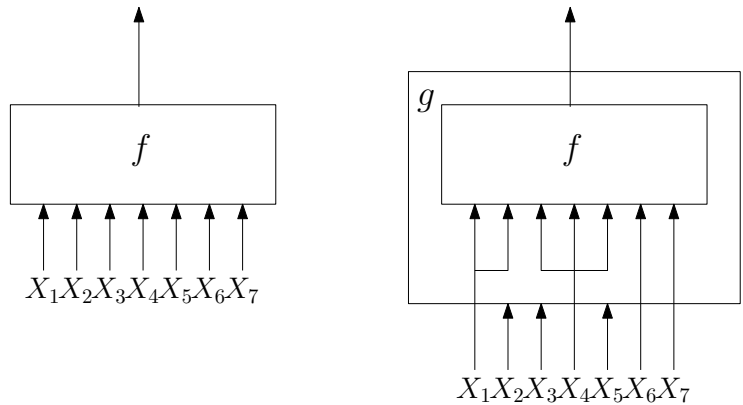, width=0.8\textwidth}
  \caption{A $7$-ary function $f$ and a gluing
    restriction $g$.}
  \label{figure-glueing}
  \end{center}
\end{figure}
The function $f'$
can be viewed as a restriction of $f$ to inputs $(a_1,\dots,a_n)$ for
which $a_i = a_j$. Thus, $f'$ can be seen as a function $A^{n-1}
\rightarrow B$. We prefer, however, to define it as a function $A^n
\rightarrow B$ that simply ignores the $j$\textsuperscript{th}
coordinate of its input.  We say $f'$ is obtained from $f$ by a {\em
  gluing step}. A function $g: A^n \rightarrow B$ is a {\em glued
  restriction} of $f$ if it can be
obtained from $f$ by a sequence of gluing steps.
See Figure~\ref{figure-glueing} for an intuition.

Consider a function $f: \{0,1\}^n \rightarrow \mathbb{R}$ and think of
feeding $f$ with random input bits. Formally, let $X_1,\dots,X_n$ be
$n$ independent binary random variables, each with expectation $p$. We
are interested in the term $\E[f(X_1,\dots,X_n)]$. In a second
scenario, we introduce dependencies between the $X_i$ by gluing some
of them together: For example, instead of choosing $X_1,\dots,X_n$
independently, we use the same bit for $X_1$, $X_2$, and $X_n$, thus
computing $\E[f(X_1,X_1,X_3,X_4,\dots,X_{n-1},X_1)]$ instead of
$\E[f(X_1,\dots,X_n)]$. With the terminology introduced above, we want
to compare $\E[f(X_1,\dots,X_n)]$ to $\E[g(X_1,\dots,X_n)]$, where $g$
is a glued restriction of $f$.  For general functions $f$, we cannot
say anything about how $\E[f(X_1,\dots,X_n)]$ compares to
$\E[g(X_1,\dots,X_n)]$. However,
if $f$ is submodular, we can.

To get an intuition, consider the boolean lattice $\{0,1\}^n$ with
$\vec{0}$ at the bottom and $\vec{1}$ at the top. In that lattice,
$\vec{x}\wedge\vec{y}$ is below $\vec{x}$ and $\vec{y}$, and
$\vec{x}\vee\vec{y}$ is above them. Thus, in some sense, the points
$\vec{x}$ and $\vec{y}$ lie between $\vec{x}\wedge\vec{y}$ and
$\vec{x}\vee\vec{y}$. See Figure~\ref{figure-lettuce} for an illustration.
\begin{figure}
\begin{center}
  \epsfig{width=0.4\textwidth,file=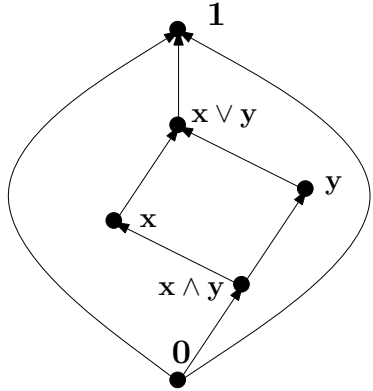}
\caption{The boolean lattice with four points $\vec{x}$,
$\vec{y}$, $\vec{x} \wedge \vec{y}$ and $\vec{x}\vee \vec{y}$.}
\label{figure-lettuce}
\end{center}
\end{figure}
On the left-hand side of
(\ref{ineq-submodular}), we evaluate $f$ at points that lie more to
the middle of the lattice, whereas on the right-hand side we evaluate
$f$ at points that lie more to the bottom or top of it.  The random
vector $(X_1,\dots,X_n)$ tends to lie around the
$pn$\textsuperscript{th} level of the lattice, whereas
$(X_1,X_1,X_3,X_4,\dots,X_{n-1},X_1)$ is less concentrated and more
often visits the extremes of the lattice. In the light of
(\ref{ineq-submodular}), we expect that biasing points towards the
extremes will decrease $\E[f]$. The following lemma formalizes this
intuition.
\begin{lemma}
  Let $f: \{0,1\}^n \rightarrow \mathbb{R}$ be a submodular function
  and $g$ be a glued restriction of it.  Let $X_1,\dots,X_n$ be
  independent binary random variables, each with expectation $p$. Then
  $\E[f(X_1,\dots,X_n)] \geq  \E[g(X_1,\dots,X_n)]$.
\label{lemma-corr}
\end{lemma}
\begin{proof}
  It is easy to see that applying a gluing step to a submodular
  function results in a submodular function: After all, a gluing step
  simply means restricting the function to a subset of its domain.
  Therefore, it suffices to prove the lemma for a function $g$ that
  has been obtained from $f$ by a single gluing step.  Without loss
  of generality, we can assume that $X_{n-1}$ and $X_n$ have been
  glued together. We have to show that
  $$
  \E[f(X_1,\dots,X_n)] \geq   \E[f(X_1,\dots,X_{n-1},X_{n-1})]\ . 
  $$
  It suffices to show this inequality for every fixed
  $(n-2)$-tuple of values for $(X_1,\dots,X_{n-2})$. Formally,
  for $b_1,\dots,b_{n-2} \in \{0,1\}$, let 
  $$
  g: (x,y) \mapsto f(b_1,\dots,b_{n-2},x,y) \ . 
  $$
  The function $g$ is also submodular.  Let $X,Y$ be two independent
  binary random variables, each with expectation $p$. We have to show
  that $\E[g(X,Y)] \geq \E[g(X,X)]$.  This is not difficult:
  \begin{eqnarray*}
    \E[g(X,Y)] & = & 
    (1-p)^2 \cdot g(0,0) + p(1-p) \cdot g(1,0) +\\
    & & + (1-p)p \cdot g(0,1) + p^2 \cdot g(1,1) \\
    & = & (1-p)^2 \cdot g(0,0) + p(1-p)\cdot (g(1,0) + g(0,1)) +
     p^2 \cdot g(1,1) \\
    & \geq &
    (1-p)^2 \cdot g(0,0) + p(1-p)\cdot(g(0,0) + g(1,1)) + p^2 \cdot 
    g(1,1) \\
    & = & ( (1-p)^2 + p(1-p))\cdot g(0,0) + (p(1-p) + p^2) \cdot g(1,1) \\
    & = & (1-p) \cdot g(0,0) + p \cdot g(1,1) = \E[g(X,X)] \ ,
    \end{eqnarray*}
    where the inequality comes from the submodularity of $g$.
\qed
\end{proof}

\begin{lemma}
  Let $I \subseteq \mathbb{R}$ be an interval, and 
  let $f: \{0,1\}^n \rightarrow I$ be monotone and submodular,
  and $h : I \rightarrow \mathbb{R}$ be non-decreasing
  and concave. Then $h \circ f:  \{0,1\}^n \rightarrow \mathbb{R}$
  is also monotone and submodular.
  \label{lemma-concave}
\end{lemma}

\begin{proof}
  It is clear that $h \circ f$, being the composition of two monotone
  functions, is again monotone. To show submodularity, consider
  $\vec{x},\vec{y} \in \{0,1\}^n$. Without loss of generality,
  $f(\vec{x}) \leq f(\vec{y})$. Using monotonicity, we see that
  $$
  f(\vec{x} \wedge \vec{y}) \leq f(\vec{x}) \leq f(\vec{y}) \leq
  f(\vec{x} \vee \vec{y}) \ .
  $$
  \textbf{Claim.} If $s \leq t$ are in $I$, and $a \geq b \geq
  0$ are such that $s- a \in I$ and $t+b \in I$, then $h(s)+h(t) \geq
  h(s-a) + h(t+b)$.\\ 
  \begin{figure}
    \begin{center}
      \epsfig{file=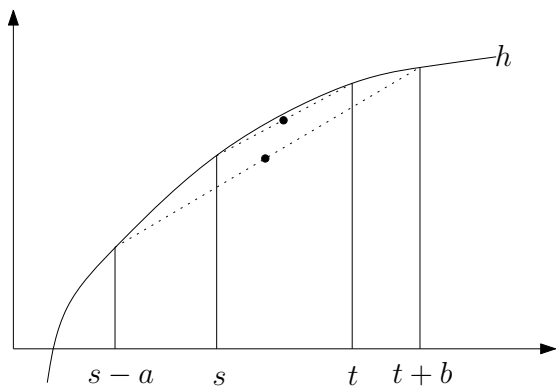, width=0.6\textwidth}
      \caption{A monotone concave function $f$ and two line
        segments.} 
      \label{figure-concave}
    \end{center}
  \end{figure}
  See Figure~\ref{figure-concave} for an illustration. To prove the
  claim, compare the line from $(s,h(s))$ to $(t,h(t))$ to the line
  from $(s-a,h(s-a))$ to $(t+b,h(t+b))$.  The midpoints of those lines
  have the coordinates
  $$
  \left(\frac{s+t}{2},\frac{h(s)+h(t)}{2}\right) \textnormal { and }
  \left(\frac{s-a+t+b}{2},\frac{h(s-a)+h(t+b)}{2}\right) \ ,
  $$
  respectively. Since $a \geq b$, the first midpoint lies to the right
  of the second midpoint. Since both lines have positive slope (by
  monotonicity of $h$) and the first line lies above the second, we
  conclude that also the first midpoint lies above the second.
  Therefore
  $(h(s-a)+h(t+b))/2 \leq (h(s)+h(t))/2$, as claimed.\\

  We apply the above claim with $s = f(\vec{x})$, $t = f(\vec{y})$, $a
  = f(\vec{x}) - f(\vec{x \wedge y})$ and $b = f(\vec{x} \vee
  {\vec{y}}) - f(\vec{y})$. Note that $s,t,s-a,t+b\in I$ and $a,b
  \geq 0$. To apply the claim we need that $a \geq b$, i.e.,
  $$
  f(\vec{x}) - f(\vec{x \wedge y}) \geq 
   f(\vec{x} \vee \vec{y}) - f(\vec{y}) \ ,
  $$
  which follows from submodularity. The claim implies
  that $h(s)+h(t) \geq h(s-a) + h(t+b)$, which with these
  particular values of $s$,$t$,$a$, and $b$ yields
  $h(f(\vec{x}))+h(f(\vec{y})) \geq
  h(f(\vec{x}\wedge\vec{y}))+h(f(\vec{x}\vee\vec{y}))$.
  \qed
\end{proof}

\ignore{
The following lemma can be verified quickly.
\begin{lemma}
  Let $f_1,\dots,f_k : \{0,1\}^n \rightarrow \mathbb{R}$ be 
  monotone and submodular. Then $f_1 + \dots + f_k$ is also
  monotone and submodular.
\label{lemma-sum-submodular}
\end{lemma}
\begin{corollary}
  The function $f: \{0,1\}^{(d-\ell)(k-1)} \rightarrow \mathbb{R}$ defined
  by
  $$
  (Y^{(\ell+1)}_1,\dots,Y^{(d)}_{k-1})
  \mapsto \log_2 \left(\ell + \sum_{c=\ell+1}^d \OR(Y^{(c)}_1,\dots,Y^{(c)}_{k-1})
  \right)
  $$
  is submodular and monotone.
\end{corollary}
\begin{proof}
  The function $\OR$ is monotone and submodular, and so is the
  constant $\ell$, viewed as a $0$-ary boolean function. By
  Lemma~\ref{lemma-sum-submodular}, the sum $\ell+ \sum_{c=\ell+1}^d
  \OR(Y^{(c)}_1,\dots,Y^{(c)}_{k-1})$ is submodular and monotone.
  Applying Lemma~\ref{lemma-concave} with the interval $I =
  [1,\infty)$, the logarithm of that sum yields again a submodular and
  monotone function, since $\log_2(\cdot)$ is nondecreasing and
  concave on $[1,\infty)$.
  \qed
\end{proof}
}

\begin{proof}[Proof of Lemma~\ref{worst-case-independent}]
  We define $(d-\ell)(k-1)$ random variables $Z_j^{(c)}$ for $1 \leq j
  \leq k-1$ and $\ell < c \leq d$. These random variables are 
  all independent and each has expectation $1-r$. We define
  the function $f: \{0,1\}^{(d-\ell)(k-1)}$ by
  \begin{eqnarray}
  f(x_1^{(\ell+1)},\dots,x^{(d)}_{k-1}) =
  \log_2 \left( \ell + \sum_{c=\ell+1}^{d}
    \OR(x_1^{(c)} \vee \dots \vee  x_{k-1}^{(c)})\right) \ .
  \label{def-of-f}
  \end{eqnarray}
  This function is clearly monotone. We claim that it is submodular:
  The $\OR$-function is submodular, and it is easy to check that a sum
  of submodular functions is again submodular. Finally, the function
  $t \mapsto \log_2(\ell + t)$ is concave.  We apply
  Lemma~\ref{lemma-concave} with the interval $I=[0,\infty)$, the
  submodular function $\sum_{c=\ell+1}^{d} \OR(x_1^{(c)} \vee \dots
  \vee x_{k-1}^{(c)})$, which has domain $I$, and the concave function
  $t \mapsto \log_2(\ell+t)$. Thus $f$ is submodular and monotone.  To
  prove Lemma~\ref{worst-case-independent}, we have to show that
  \begin{eqnarray}
    \E\left[\log_2\left(\ell + \sum_{c=\ell+1}^d
        Y^{(c)}\right) \ | \ \pi(x)=r\right] \leq \E\left[\log_2\left(\ell +
        \sum_{c=\ell+1}^{d} Z^{(c)}\right)\right] \ ,
    \label{ineq-worst-case-indep}
  \end{eqnarray}
  where the $Z^{(c)}$ are independent binary random variables with
  expectation $1 - r^{k-1}$ and 
  $Y^{(c)} := \OR(Y_1^{(c)}, \dots, Y_{k-1}^{(c)})$, with
  \begin{eqnarray*}
    Y_j^{(c)} := \left\{ \begin{array}{ll}
        1 &  \textnormal{ if } y_j^{(c)} \textnormal{ comes after } x 
        \textnormal{ in the permutation } \pi \ ,\\ 
        0 &  \textnormal{otherwise} \ .
      \end{array}
    \right.
  \end{eqnarray*}
  The left-hand side of (\ref{ineq-worst-case-indep}) thus reads
  as 
  $$
  \E[f(Y_1^{(\ell+1)},\dots,Y^{(d)}_{k-1} \ | \ \pi(x)=r] 
  $$
  for $f$ as defined in (\ref{def-of-f}). Since the
  $Z^{(c)}$ are independent binary random variables with expectation
  $1 - r^{k-1}$, their dis\-tri\-bu\-tion is identical to the
  distribution of
  $\OR(Z_1^{(c)},\dots,Z_{k-1}^{(c)})$, and the right-hand side of
  (\ref{ineq-worst-case-indep}) is equal to
  $$
  \E[f(Z_1^{(\ell+1)},\dots,Z^{(d)}_{k-1}] \ .
  $$
  We have to show that
  \begin{eqnarray}
  \E[f(Y_1^{(\ell+1)},\dots,Y^{(d)}_{k-1} \ | \ \pi(x)=r] \leq 
  \E[f(Z_1^{(\ell+1)},\dots,Z^{(d)}_{k-1}]
  \label{ineq-worst-case-indep-2}
  \end{eqnarray}
  Conditioned on $\pi(x)=r$, the distribution of each $Y_j^{(c)}$ is
  identical to that of $Z_j^{(c)}$, but some $Y_j^{(c)}$ are ``glued
  together'', since the underlying variables $y_j^{(c)}$ of our CSP
  formula need not be distinct. We can, however, assemble the
  $Y_j^{(c)}$ into groups according to their underlying variables
  $y_j^{(c)}$ such that (i) random variables from the same group have
  the same underlying $y_j^{(c)}$ and thus are identical, (ii) random
  variables from different groups are independent. Thus,
  $f(Y_1^{(\ell+1)},\dots,Y^{(d)}_{k-1}$ is a glued restriction of
  $f(Z_1^{(\ell+1)},\dots,Z^{(d)}_{k-1}$ or rather can be coupled with
  a glued restriction thereof, and thus by Lemma~\ref{lemma-corr}, the
  expectation of the former is at most the expectation of the latter.
  Therefore (\ref{ineq-worst-case-indep-2}) holds.  \qed
\end{proof}

\ignore{
\section{Concluding Remarks}
For $d=2$, the algorithm described above is exactly the algorithm
$\ppz$ as stated by Paturi, Pudl\'ak and Zane~\cite{feder-motwani}, and our
bound on the success probability coincides with theirs. For $d \geq 3$
and $k=2$, our algorithm equals the one by Feder and
Motwani~\cite{feder-motwani}. However, for larger values of $k$, the critical
constraints $C_{\ell+1},\dots,C_d$ in (\ref{critical-clauses}) can
overlap in complicated patterns, and it is not clear a priori which
pattern constitutes the worst case.  Luckily, one can prove that the
worst case happens when the constraints share no variables besides
$x$.
}
\bibliographystyle{abbrv}
\bibliography{refs}

\end{document}